\documentclass[12pt,reqno,centertags]{amsart}
\usepackage{a4}
\usepackage{amsmath,amsthm,amscd,amssymb}
\usepackage{latexsym}
\usepackage{upref}
\usepackage{hyperref}
\usepackage{upgreek}
\usepackage{textgreek}
\usepackage{cite}
\usepackage{tikz}
\usepackage{enumerate}
\usepackage{multicol}
\usepackage{accents}
\usepackage{lmodern}
\usepackage{microtype}
\usepackage{graphicx}
\usepackage{subcaption}
\usepackage{dsfont}
\usepackage{colonequals}
\UseRawInputEncoding

\newcommand{\abs}[1]{\left|#1\right|}
\newcommand\rE[1]{\upharpoonleft_{#1}}
\newcommand\lrb[1]{\left\lbrace#1 \right\rbrace}
\newcommand\lrp[1]{\left(#1 \right)}
\newcommand\C{{\mathbb C}}

\newcommand\E{{\mathbb E}}

\newcommand{\ceq}{\colonequals}
\newcommand\cc[1]{\overline {#1}}
\newcommand\vb[1]{\!\left\langle {#1} \right |}
\newcommand\vk[1]{\left |{#1}\right\rangle\!}
\newcommand\ip[2]{\left\langle {#1},{#2} \right\rangle}
\newcommand\no[1]{\left\| {#1} \right\|}
\newcommand{\initsp}{\mathsf{h}}
\newcommand\unit{\hbox{\rm 1\kern-2.8truept l}}
\newcommand\qF[1]{\mathcal F[#1]}
\newcommand\tr[1]{{{\rm tr}}\left(#1\right)}
\DeclareMathOperator{\im}{Im}
\DeclareMathOperator{\re}{Re}
\DeclareMathOperator{\dom}{dom}
\DeclareMathOperator{\ran}{ran}
\DeclareMathOperator{\ad}{ad}

\allowdisplaybreaks
\numberwithin{equation}{section}

\newtheorem{theorem}{Theorem}[section]
\newtheorem{lemma}[theorem]{Lemma}
\newtheorem{corollary}[theorem]{Corollary}

\newtheorem{proposition}[theorem]{Proposition}
\theoremstyle{definition}
\newtheorem{definition}[theorem]{Definition}
\newtheorem{example}[theorem]{Example}

\theoremstyle{remark}
\newtheorem{remark}[theorem]{Remark}
\newtheorem*{acknowledgments}{Acknowledgments}

\begin{document}
\title[Quantum Gaussian States in Terms of Annihilation Moments]{A Characterization of Quantum Gaussian States in Terms of Annihilation Moments}
\author[J.R. Bola\~nos-Serv\'in, R. Quezada and J. I. Rios-Cangas]{Jorge R. Bola\~nos-Serv\'in, Roberto Quezada and Josu\'e I. Rios-Cangas}
\address{Departamento de Matem\'aticas, Universidad Aut\'onoma Metropolitana, Iztapalapa Campus,  San Rafael Atlixco 186, 09340 Iztapalapa, Mexico City.}
\email{jrbs@xanum.uam.mx,\quad roqb@xanum.uam.mx,\quad jottsmok@xanum.uam.mx}

\date{\today}

\begin{abstract}
We give a rigorous definition of moments of an unbounded observable with respect to a quantum state in terms of Yosida's approximations of unbounded generators of contractions semigroups. We use this notion to characterize Gaussian states in terms of annihilation moments. As a by-product, rigorous formulae for the mean value vector and the covariance matrix of a Gaussian state are obtained.
\end{abstract}

\maketitle

\section{Introduction}\label{s1}
Quantum Gaussian states, also called quasi-free states, are the natural extension to the quantum setting of the notion of Gaussian or normal distributions in classical probability theory. They are intensively studied in Quantum Physics together with quantum Gaussian channels and semigroups. Surprisingly, mathematical approaches on the topic are scarce and, even though there has been an increasing interest in both the finite and the infinite dimensional cases \cite{MR2662722,MR2683408,MR464938,MR551128,MR4218302,MR4065489}, there are still several notions regarding moments of unbounded observables and their relationship with gaussianity that require a mathematically rigorous approach.  

The main aim of this work is to develop a mathematically rigorous framework to the notion of moments of an unbounded observable with respect to a quantum state and provide a characterization of Gaussian states in terms of annihilation moments, i.e., moments of the Weyl generator $-\frac{1}{2i}(za^\dagger-\cc za)=\im\lrp{\cc{z}a}$, which we denote by $\sigma(z,a)$, along the lines of the underlying standard symplectic structure, see Theorem \ref{characterization-gaussian-state} and Corollary \ref{necessity} below. A bilinear form $\sigma\colon \initsp\times \initsp \mapsto \mathbb R$, where $\initsp$ is a Hilbert space, is called \emph{symplectic} if $\sigma(u, v)=-\sigma(v, u)$ for all $u, v\in \initsp$ and the pair $(\initsp, \sigma)$ is called a \emph{symplectic space}. The space is called \emph{standard} if $(\initsp,\ip \cdot \cdot)$ is a complex Hilbert space and $\sigma(u, v)= \im\ip uv$. We consider only the one dimensional case, i.e., one oscillation mode, the multi and infinite dimensional cases will be considered in a forthcoming work. 

After a brief introduction to the one-mode coherent states representation of the CCR in Section \ref{s2}, we define in Section \ref{annihilation-moments} the $n$-th moment of an unbounded observable $A$ in a quantum state, by means of Yosida's approximations of the generator $iA$ of a strongly continuous unitary group; and we show the relationship of our approach with that of Krauss and Schr\"{o}ter \cite{MR363279}. Moreover, we prove that all annihilation moments in a Gaussian state are finite and satisfy the recurrence relation \eqref{g-moments}. Section \ref{characterization-gaussian} is devoted to prove that the above property characterizes Gaussian sates, Theorem \ref{characterization-gaussian-state}. As a by-product we provide in Corollary \ref{mean-covariance}, a mathematically rigorous proof of formulae for the mean value vector and the covariance matrix of Gaussian states, which are freely used in the literature without a proof. 

\section{Preliminaries}\label{s2}
\subsection{Coherent-state representation of CCR}\label{subsec:Coherent-state-rep}

The \textit{coherent-state representation of CCR} take place on the subspace $\mathcal E_2(\C)\subset L_2(\C)$ with orthonormal basis (onb for short) $\{\phi_k\}_{k\geq0}$ (see Theorem \ref{K-B-Iso}), where
\begin{gather*}
\phi_k(z) = e^{-\frac{|z|^2}{2}}\frac{\cc z^k}{\sqrt{\pi k!}}\,.
\end{gather*}
To obtain this representation, we start with the \emph{toy Fock space} $\Gamma_{s}(\C)$, which is isomorphic to $L_{2}(\mathbb R)$ via the identification
\begin{align*}\varepsilon(z)\longmapsto f_z(x)=e^{\sqrt{2}zx-\frac{z^2}{2}}\unit(x)\,,\quad \varepsilon\in\Gamma_s(\C)
\end{align*}
where $\unit(x)=\pi^{-\frac14}e^{-\frac{x^2}{2}}$ is well-known as the \emph{ground state} and $f_z$ as the \emph{exponential} vectors. Note that $\langle f_z,f_{z'}\rangle= e^{\overline{z}z'}$.

Let $Qf(x)=xf(x)$ and $Pf(x)=-i\frac{df(x)}{dx}$ be realisations of position and  momentum operators on $L_{2}({\mathbb R})$ which satisfy $[Q,P]=i I$.

Creation and annihilation operators on $L_{2}(\mathbb R)$ are defined by 
\begin{gather*}
a^\dagger \ceq \frac{1}{\sqrt{2}}(Q - iP)\quad;\quad a\ceq \frac{1}{\sqrt{2}}(Q + i P)\,.
\end{gather*}
It follows that  $Q=(a^\dagger+a)/\sqrt{2}$, $P=i(a^\dagger -a)/\sqrt{2}$ and $[a,a^\dagger]=I$. Besides, the Hamiltonian of the quantum harmonic oscillator is 
\begin{gather*}
H=N+\frac{1}{2}I= \frac{1}{2}(Q^{2} + P^{2})= \frac{1}{2}\lrp{x^2 -  \frac{d^2}{dx^{2}}} \,,\end{gather*}
where $N=a^\dagger a$ is the number operator. 

\begin{definition}
For $z=r + i s\in\C$ and $f\in L_{2}(\mathbb R)$, the \emph{Weyl (or displacement)} operator $W_z$ is given by 
\begin{gather}\label{eq:Weyl-operator-initial}
W_{z}f(x)\ceq e^{-i\sqrt{2}(r P -s Q)}f(x)\,.
\end{gather}
\end{definition}
The Weyl operator satisfies the so called \emph{Schr\"odinger representation}
\begin{gather}\label{Weyl-operators}
W_{z}f(x)=e^{-is(r-\sqrt{2}x)}f(x-\sqrt{2}r)\,.
\end{gather}
Indeed, from the Baker-Campbell-Hausdorff formula 
\begin{gather*}
e^{A+B}= e^{B}e^{A}e^{\frac{1}{2}[A,B]}\,,
\end{gather*} with $A=-irP$, $B=i sQ$ and $[A,B]= rs[P,Q]=-irs I$, we get \eqref{Weyl-operators}. We highlight that \eqref{Weyl-operators} is  irreducible (i.e., the only invariant subspaces are $\{0\}$ and the whole $L_2(\mathbb R)$) and unitary, with adjoint $W_z^*=W_{-z}$. Besides, 
\begin{gather}\label{Weyl-CCR}
W_{z} W_{z'}=e^{-i(rs'-r's)}W_{z+z'}=e^{-i\sigma(z, z')}W_{z + z'}\,,
\end{gather}
where $\sigma(z,u)=\im (\cc zu)$ is the \emph{standard symplectic form}. Moreover, \eqref{eq:Weyl-operator-initial} implies 
\begin{align}\label{eq:Weyl-a-a-dagger}
\begin{split}
W_{z}= e^{-i\lrp{ir(a^{\dagger}-a)-s(a^{\dagger}+ a)}}= e^{za^{\dagger}-\cc za} = e^{-\frac{\abs z^2}{2}} e^{za^{\dagger}}e^{-\cc za}\,.
\end{split}
\end{align}

We consider the family of \emph{coherent} vectors $\psi_z\ceq W_z \unit\in L_2(\mathbb R)$, $z\in\C$. By virtue of  $a \unit=0$ and \eqref{eq:Weyl-a-a-dagger}, the coherent vectors satisfy
\begin{align}\label{eq:coherent-v-p}
\begin{split}
\psi_z (x)&=e^{-\frac{\abs z^2}{2}}\sum_{k\geq 0}\frac{z^k}{k!}\frac{(Q-iP)^k}{\sqrt{2^k}}\unit (x) =e^{-\frac{\abs z^2}{2}}\sum_{k\geq 0}\frac{\lrp{z/\sqrt{2}}^k}{k!}h_k(x) \unit (x)\\&=  e^{-\frac{\abs z^2}{2}}e^{\sqrt{2}xz-\frac{z^2}{2}} \unit(x)=e^{-\frac{\abs z^2}{2}}f_z(x) \,,
\end{split}
\end{align}
being $h_k$ the (physicist') Hermite polynomials (orthogonal w.r.t. the $e^{-x^2}$ density and with generating function $e^{2z x - z^2}$)  $1, \, 2x, \, 4x^2 - 2, \, 8x^3 -12x,\dots$  On account of \eqref{eq:coherent-v-p}, the family of coherent vectors $\psi_z$ coincide with the normalized exponential vectors $f_z$, but they are not orthogonal, since $\langle \psi_{z}, \psi_{z'}\rangle = e^{-\frac{1}{2}(|z|^2 + |z'|^2 -2\overline{z} z')}$.

We also consider an onb for $L_2(\mathbb R)$, composed of the set of \emph{Hermite functions},
\begin{gather*}
\lrb{\varphi_k\ceq (2^k k!)^{-\frac{1}{2}} \initsp_{k}\unit}_{k\geq 0}\,.
\end{gather*}
One can check that $a^\dagger \varphi_k  =\sqrt{k+1} \varphi_{k+1}$ and $a \varphi_k =\sqrt{k} \varphi_{k-1}$. Besides,
 \begin{gather}\label{eq:onb-basis}
 \psi_z=e^{-\frac{\abs z^2}{2}}\sum_{k\geq 0}\frac{z^k}{ \sqrt{k!}} \varphi_k\quad\mbox{;}\quad W_{z'}\psi_{z}=e^{-i\sigma(z', z)}\psi_{z' + z}\,,
 \end{gather}
whence it follows that $a\psi_z=z\psi_z$, i.e., the coherent vectors are a continuum of eigenvectors of the annihilation operator.

\begin{lemma}
The completeness relation
\begin{gather}\label{eq:completeness-rel}
\frac{1}{\pi}\int\vk{W_zf}\vb{W_zf}dz=I\,, 
\end{gather}
holds true, for any unit vector $f\in L_2(\mathbb R)$.
\end{lemma}
\begin{proof} 
Indeed, following \cite[Prop.\,3.5.1]{MR2797301} (rescaled by $(2)^{-1/2}$), if $\{e_k\}_{k\geq0}$ is a onb  for $L_2(\mathbb R)$, then  
\begin{gather*}
\frac{1}{\sqrt{\pi}}\lrb{\ip{e_j}{W_{z}\,e_k}}_{j,k\geq0}
\end{gather*}
is  an onb for $L_2(\mathbb R^2)$ and the orthogonality relations hold
\begin{gather*}
\frac{1}{\pi}\int_{\C}\cc{\ip{e_j}{W_{z}\,e_k}}{\ip{e_l}{W_{z}\,e_m}}dz=\delta_{jl}\delta_{km}\,.
\end{gather*}
In this fashion, for $ f,g,h,k\in L_2(\mathbb R)$, by linearity, 
\begin{gather}\label{eq:orthogonality-condition}
\frac{1}{\pi}\int_{\C}\ip{h}{W_{z}g}\ip{W_{z}f}{k}dz=\ip{f}{g}\ip{h}{k}\,.
\end{gather} 
Hence, due to $h,k$ are arbitrary and putting $f=g$ with unit norm, we conclude that \eqref{eq:completeness-rel} is true, with the integral defined in weak sense.
\end{proof}

On account of \eqref{eq:completeness-rel}, the formula $\varphi = \int\ip{\psi_z}{\varphi} \psi_z dz/\pi$ suggest considering a map $\varphi\mapsto\ip{\psi_z}{\varphi}/\sqrt{\pi}$ from $L_{2}({\mathbb R})$ into functions  $\phi(z)=\ip{\psi_z}{\varphi}/\sqrt{\pi}$. We denote by ${\mathcal E}_{2}(\C)$ the image of $L_2({\mathbb R})$ under this map, which is sometimes called \textit{Klauder-Bargman Isomorphism} (c.f. \cite{MR4218302}).

\begin{theorem}\label{K-B-Iso} The map $\varphi\mapsto \phi(z)=\ip{\psi_z}{\varphi}/\sqrt{\pi}$ from $L_{2}({\mathbb R})$ into $L_2(\C)$ is an isometry and the family of the \emph{canonical coherent-states}
\begin{gather}\label{eq:canonical-coherent-states}
\lrb{\phi_k(z)\ceq\ip{\psi_z}{\varphi_k}/\sqrt{\pi}}_{k\geq0}\,,
\end{gather} is an onb for ${\mathcal E}_{2}(\C)$.
\end{theorem}
\begin{proof}
For any pair of elements $\varphi, \varphi' \in L_{2}(\mathbb R)$, one obtains from \eqref{eq:orthogonality-condition} that 
\begin{align}\label{eq:isometric-equivalence}
\begin{split}
\ip{\phi}{\phi'}=\frac{1}{\pi}\int_{\C} \overline{\ip{\psi_z}{\varphi}} \ip{\psi_z}{\varphi'} dz= \frac{1}{\pi}\int_{\C}\ip{\varphi}{\psi_z}\ip{\psi_z}{\varphi'}dz=\ip{\varphi}{\varphi'}\,.
\end{split}
\end{align}
Since $\{\varphi_k\}_{k\geq0}$ is a onb for $L_2(\mathbb R)$, the family \eqref{eq:canonical-coherent-states} satisfies that 
\begin{align*}
\int_{\C}\cc{\phi_j(z)}\phi_k(z)dz=\ip{\varphi_j}{\varphi_k}=\delta_{jk}\,,
\end{align*}
which implies that \eqref{eq:canonical-coherent-states} is an orthonormal basis for ${\mathcal E}_2(\C)$. 
\end{proof}
By virtue of \eqref{eq:onb-basis} the canonical coherent-states satisfy
\begin{gather*}
\phi_k(z)=e^{-\frac{|z|^2}{2}}\frac{\cc z^k}{\sqrt{\pi k!}}\,,\quad k\geq0\,.
\end{gather*}
 
\begin{remark}
Any bounded operator $X$ in  $L_{2}(\mathbb R)$ corresponds with a  bounded integral operator $\hat X$  in ${\mathcal E}_2(\C)$, with kernel $\ip{\psi_z}{X\psi_v}/\pi$ and $\|\hat X\|=\no{X}$. Indeed, if $\hat X$ maps $\phi(z)= \ip{\psi_z}{\varphi}/\sqrt{\pi}$ to $\ip{\psi_z}{X\varphi}/\sqrt{\pi}$, then from \eqref{eq:completeness-rel}, 
\begin{gather*}
\no{\hat X\phi}^2=\frac1\pi\int_\C\ip{X\varphi}{\psi_z}\ip{\psi_z}{X\varphi}dz=\ip{X\varphi}{X\varphi}=\no{X\varphi}^2\,,
\end{gather*}
which implies  $\|\hat X\|=\no{X}$, since $\no{\phi}=\no{\varphi}$ (v.s. \eqref{eq:isometric-equivalence}). Besides, 
\begin{align*}
\hat X\phi(z)=\frac{\ip{\psi_z}{X\varphi}}{\sqrt{\pi}}=\frac{1}{\pi}\int_\C\ip{\psi_z}{X\psi_v}\frac{\ip{\psi_v}{\varphi}}{\sqrt{\pi}}dv=\int_{\C}\phi(v)\frac{\ip{\psi_z}{X\psi_v}}{\pi}dv\,,
\end{align*}
as required. In particular, by \eqref{eq:onb-basis}, the weyl operator in ${\mathcal E}_2(\C)$ has kernel  
\begin{align}\label{eq:Weyl-kernel}
\frac{\ip{\psi_u}{W_z\psi_v}}{\pi}=\frac1{\pi}e^{-\frac{\abs{z+v-u}^2}{2}+i\sigma(u,z+v)-i\sigma (z,v)}\,.
\end{align}
\end{remark}

\subsection{Non-commutative Fourier Transform}\label{QFT}

Let $\initsp$ be a separable Hilbert space, particularly  $\initsp=L_2(\mathbb R)$, and denote by $L_{1}(\initsp)$ and $L_{2}(\initsp)$ the Banach and Hilbert spaces of finite-trace and Hilbert-Schmidt operators with norm $\no{\rho}_1=\tr{|\rho|}$ and inner product $\ip{\rho}{\eta}_2 = \tr{\rho^* \eta}$, respectively. 

We  define the \emph{quantum (or non-commutative) Fourier transform} of a trace-class operator $\rho$ in $L_{1}\lrp{\initsp}$, by means of  
\begin{gather*}
\qF\rho(z)\ceq\frac{1}{\sqrt{\pi}} \tr{\rho W_z}, \quad z\in\C
\end{gather*}
which is a complex valued function on $\C$, with the following properties:
\begin{align*}
\abs{\qF\rho(z)}&\leq\no{\rho}_1/\sqrt{\pi}&\qF{\rho W_u}(z)&=e^{-i\sigma(u,z)}\qF\rho(z+u)\\
\qF{\rho^*}(z)&=\cc{\qF\rho(-z)}&\qF{W_u^*\rho W_u}(z)&=e^{-2i\sigma(u,z)}\qF\rho(z)
\end{align*}

It is worth pointing out that the map $\rho \mapsto\qF\rho$ extends uniquely to a unitary map from $L_{2}\lrp{\initsp}$ onto $L_2(\C)$. Thence, one has the \emph{Non-commutative Parseval identity} 
\begin{gather}\label{parceval-id}
\int\cc{\qF{\rho}(z)}\qF{\eta}(z)dz=\tr{\rho^* \eta}\,,\quad  \rho,\eta \in L_{2}\lrp{\initsp}\,.
\end{gather}
Besides, a state $\rho$ is pure if and only if  $\no{\qF\rho}=1$(c.f. \cite{MR2797301}).

The \emph{Weyl transform} of a vector $\phi\in L_1(\C)$ is defined by
\begin{gather*}
W(\phi)\ceq \frac{1}{\sqrt{\pi}}\int \phi(z)W_{-z}dz\,,
\end{gather*} 
which is a well-defined operator in ${\mathcal B}(\initsp)$, due to the integral is in weak sense (also as a Bochner integral). Indeed, from \eqref{eq:Weyl-kernel}, for any coherent vector $\psi_{u}\in \initsp$,
\begin{align*}
\abs{\ip{\psi_u}{W(\phi)\psi_u}}&\leq\frac{1}{\sqrt{\pi}}\int\abs{\phi(z)\ip{\psi_u}{W_{-z}\psi_u}}dz=\frac{1}{\sqrt{\pi}}\int\abs{\phi(z)e^{2i\sigma(u,z)}e^{-\frac{\abs{z}^2}{2}}}dz\\&=\frac{1}{\sqrt{\pi}}\int\abs{\phi(z)}e^{-\frac{\abs{z}^2}{2}}dz<\infty\,,
\end{align*}
since $\|{e^{-{\abs{z}^2}/{2}}}\|_\infty=1$. Thus, via polarization identity for sesquilinear forms, one can conclude that the integral exists in weak sense. Besides, the correspondence $\phi\mapsto W(\phi)$ is one-to-one, viz. for every $f,g\in \initsp$,
\begin{align*}
0&=\sqrt{\pi}\ip{f}{W_{-u/2}W(\phi) W_{u/2}g}=\int\phi(z)\ip{f}{W_{-u/2}W_{-z} W_{u/2}g}dz\\&=\int e^{-i\sigma(u,z)}\phi(z)\ip{f}{W_{-z}g}dz\,.
\end{align*}
Hence, the uniqueness of the usual Fourier transform yields $\phi(z)\ip{f}{W_{-z}g}=0$, i.e., $\phi(z)=0$ a.e. on $\C$.
\begin{corollary}
If $\rho, \phi \in L_2(\initsp)$, then 
\begin{gather*}
W(\qF\rho)=\rho\quad \mbox{and}\quad \qF{W(\phi)}=\phi\,.
\end{gather*}
\end{corollary}
\begin{proof} 
For arbitrary $\varphi, \psi \in \initsp$, we have by Parseval identity \eqref{parceval-id} that
\begin{align*}
\tr{\vk{\varphi}\vb{\psi}^*\rho}&=\int\cc{\qF{\vk{\varphi}\vb{\psi}}(z)}\qF{\rho}(z)dz=\frac{1}{\sqrt{\pi}}\int\tr{\vk{\psi}\vb{\varphi}W_{-z}}\qF{\rho}(z)dz\\&=\frac{1}{\sqrt{\pi}}\int\qF{\rho}(z)\ip{\varphi}{W_{-z}\psi}dz\,,
\end{align*}
whence it follows that $\ip{\varphi}{\rho\psi}=\ip{\varphi}{W(\qF\rho)\psi}$, i.e., $W(\qF\rho)=\rho$. Now, since $\mathcal F$ is unitary from $L_2(\initsp)$ onto  $ L_2(\C)$, there exists $\eta\in L_2(\initsp)$ such that $\qF{\eta}=\phi$.  In this fashion, $\qF{W(\phi)}=\qF{W(\qF{\eta})}=\qF{\eta}=\phi$, as required.
\end{proof}
\section{Quantum Gaussian states and annihilation moments}\label{annihilation-moments}

We consider the real inner product $(z, z')\ceq\re\ip{z}{z'}$ on $\C$, which produces a real Hilbert space with orthonormal basis $\{1, i\}$. 
\begin{definition}
A state $\rho \in L_1(\initsp)$ is \emph{Gaussian} if there exists $w\in\C$ and a symmetric real matrix $S\in{\mathcal B}_{\mathbb R}({\C})$ such that 
\begin{gather*}
\qF\rho(z)=\frac{1}{\sqrt{\pi}}e^{-i(w,z)-\frac12 (z,Sz)}\,,\quad \forall z\in\C
\end{gather*}
in which case we write $\rho=\rho(w,S)$.
\end{definition}
The last definition determines a real linear functional $z \mapsto (w, z)$ and a real quadratic form $z \mapsto (z, Sz)$ on the real Hilbert space ${\mathbb C}$. Hence the mean value vector $w$ and the covariance operator $S$ are uniquely determined by the definition. 
\begin{remark}\label{rm:properties-wS}
If $w= \sqrt{2}(l - im)$, then we call $l$ the mean momentum vector and $m$ the mean position vector, respectively. Besides, the covariance matrix $S$ has the matrix representation 
\begin{gather*}
S=\begin{pmatrix}
 (1, S1) & (1, S i) \\ 
(i, S 1) & (i, S i) 
\end{pmatrix}\,,
\end{gather*}
 with off-diagonal  elements satisfying $(1, S i) = (i, S 1)$. 
Moreover, for $z=x+iy$
\begin{align*}
(z, Sz) = \lrp{x+iy, S(x+iy)}= (1, S1) x^2 + (i, Si) y^2 + 2(1, Si)xy\,.
\end{align*}
\end{remark}

\begin{example}[Vacuum state]\label{ex:vacuum-satate} The simplest Gaussian state is the vacuum state $\Omega=\vk{\unit}\vb{\unit}$, where $\unit$ is the ground state. Indeed, 
\begin{align*}
\sqrt{\pi}\qF{\Omega}(z)&=\tr{\Omega W_z}=\tr{W_{\frac{z}{2}}\vk{\unit}\vb{\unit} W_{\frac{z}{2}}}\\&=\ip{\psi_{\frac{z}{2}}}{\psi_{-\frac{z}{2}}}=e^{-\frac{1}{2}|z|^2}\,. 
\end{align*}
Therefore, $\Omega=\Omega(0,I)$.
\end{example}

\begin{example}[Gibbs state at inverse temperature $\beta$]\label{gibbs-state}
For $\beta>0$, the state $\rho_{\beta} = (1 -e^{-\beta}) e^{-\beta a^{\dagger}a}$ is well defined since 
\begin{gather*}
\tr{e^{-\beta a^{\dagger}a}} = \sum_{n\geq 0} e^{-n \beta} = \frac{1}{1 - e^{-\beta}}\,.
\end{gather*}
Thus,  
\begin{align*}
\frac{\sqrt{\pi}\qF{\rho_s}(u)}{1-e^{-\beta}}
&=\tr{e^{-\beta a^{\dagger}a}W_u\frac1{\pi}\int\vk{\psi_z}\vb{\psi_z}dz}\\
&=\frac1{\pi}\int\ip{\psi_z}{e^{-\beta a^{\dagger}a}W_u\psi_z}dz\\
&=\frac1{\pi}\int e^{-\frac{\abs z^2}2-i\im \cc{u}z-\frac{\abs{u+z}^2}2}\ip{\psi_z}{e^{- \beta a^{\dagger}a}\psi_{u+z}}dz\\
&=\frac1{\pi}\int e^{-\abs z^2-\cc uz-\frac{\abs{u}^2}2}\ip{\psi_z}{\psi_{e^{-\beta}(u+z)}}dz\\
&=\frac1{\pi} \int e^{-\frac{\abs u^2}{2}-(1-e^{-\beta}) \abs z^2+(\cc zu e^{-\beta}  - z \cc{u})} dz\\
&=e^{-\frac{|u|^2}{2}\lrp{1+\frac{2e^{-\beta}}{1-e^{- \beta}}}}\frac{1}{\pi} \iint e^{-(1-e^{-\beta})
\left[\lrp{x-\frac{u e^{-\beta}-\cc u}{2(1-e^{-\beta})}}^2+\lrp{y+i\frac{u e^{-\beta}+\cc u}{2(1-e^{-\beta})}}^2\right]}dxdy\\ &= \frac{1}{1- e^{-\beta}}e^{-\frac{1}{2} \textrm{coth}(\frac{\beta}{2})|u|^2}\,,
\end{align*} 
 whence $\qF{\rho_{\beta}}(u)=e^{-\frac{1}{2} \textrm{coth}(\frac{\beta }{2})|u|^2}/\sqrt{\pi}$. We have used that $e^{- \beta a^{\dagger}a}\psi_z = \psi_{e^{-\beta}z}$. Therefore, with $S=\textrm{coth}(\frac{\beta}{2}) I$ and $w=0$,  one has that $\rho_{\beta}=\rho_{\beta}(0,S)$ is Gaussian. 
\end{example}  
One can produce more examples of Gaussian states using quantum Gaussian channels, see Section \ref{channels-semigroups} below.

Recall that $\lrp{W_{tz}}_{t\geq 0}$ is a unitary group for each fixed $z\in{\mathbb C}$. Indeed,
\begin{gather*}
W_{tz} W_{sz} = e^{-i\sigma(tz, sz)}W_{tz + sz} = e^{-i ts \sigma(z, z)}W_{(t + s)z}= W_{(t+s)z}\,.
\end{gather*}
Thereby, there exists a selfadjoint operator (the infinitesimal generator) $-L_z$ such that $W_{tz}=e^{-it L_z}$. Nevertheless, bearing in mind \eqref{eq:Weyl-a-a-dagger} and the annihilation observable $\sigma(z,a)$ (q.v. Definition~\ref{def:annihilation-moments}), it follows that \begin{gather}\label{eq:symplectic-exponential-representation}
W_z=e^{za^\dagger-\cc za}=e^{-2i\sigma(z,a)}\quad (z\in\C)
\end{gather}
is not a semigroup in the complex variable $z$, since  by virtue of \eqref{Weyl-CCR}, 
\begin{gather*}
e^{-2i\sigma(z+z',a)}= e^{i \sigma(z, z')}e^{-2i\sigma(z,a)}e^{-2i\sigma(z',a)}\,.
\end{gather*}

For $z=x+iy\in\C$, it is possible to define rigorously \eqref{eq:symplectic-exponential-representation} in terms of the unitary groups  
\begin{gather*}
e^{-i\sqrt{2}xP}=e^{x (a^{\dagger}-a)} \quad\mbox{and}\quad
e^{i\sqrt{2}yQ}= e^{iy(a^{\dagger}+a)}\,.
\end{gather*}
Indeed, by virtue of \eqref{eq:Weyl-a-a-dagger} and using the Baker-Campbell-Hausdorff formula: 
\begin{align*}
e^{-2i\sigma(z,a)}= e^{-i\sqrt{2}(xP - yQ)}=e^{ixy}e^{-i\sqrt{2}xP}  e^{i\sqrt{2}yQ}\,.
\end{align*}

In what follows, we will show that all moments of an observable $A$ in a state $\rho$ can be computed by using derivatives of the functions $\tr{\rho e^{-t(iA)^n}}$, which involve $\tr{\rho A^{n}}$, $n\in\mathbb N$. However, the latter requires a rigorous definition since the observable $A$ may be unbounded. To do so, we will use the following well-know properties of Yosida's approximations for unbounded selfadjoint operators \cite[Sec.\,1.3]{MR710486}.

For $\epsilon>0$ and an infinitesimal generator $\Lambda$ of a
strongly continuous semigroup of contractions $e^{-t\Lambda}$, with $t\geq0$, one has that $(I+\epsilon \Lambda)^{-1}$ is a contraction in $\mathcal B(\initsp)$ and the operator $\Lambda_{\epsilon}\ceq \Lambda(I + \epsilon \Lambda)^{-1}\in\mathcal B(\initsp)$. Besides, for $u\in \initsp$,
\begin{gather}\label{eq:Yosida-resolvent-onH}
 \lim_{\epsilon\to 0} \lrp{I+\epsilon \Lambda}^{-1}u=u\quad\mbox{;}\quad \lim_{\epsilon\to 0} e^{-t\Lambda_{\epsilon}}u = e^{-t\Lambda}u\,,
\end{gather}
while for $u\in \dom \Lambda$,
\begin{gather}\label{eq:Yosida-resolvent-dense}
(I+\epsilon \Lambda)^{-1}\Lambda u=\Lambda(I+\epsilon \Lambda)^{-1}u\quad\mbox{;}\quad\lim_{\epsilon\to 0} \Lambda_{\epsilon}u = \Lambda u\,.
\end{gather}
Moreover, $\Lambda_{\epsilon}$ is the infinitesimal generator of the uniformly continuous semigroup of contractions $e^{-t\Lambda_{\epsilon}}$. In particular, the above reasoning holds for $\Lambda=\pm iA$, where $A$ is a selfadjoint operator. 

\begin{definition}\label{def:annihilation-moments} For a state $\rho$ and $n\in\mathbb N$, the \emph{$n$-th moment} of an observable $A$ in the state $\rho$, is defined by 
\begin{align*}
\langle A^n\rangle_\rho\ceq(-i)^n\lim_{\epsilon\to 0}\tr{\rho\big((iA)_{\epsilon}\big)^n}\,,
\end{align*}
whenever the limit exists.  Particularly for $z\in\C$, we call $\sigma(z,a)=(\cc za-za^\dagger)/2i$  the \emph{annihilation observable} and for a Gaussian state $\rho$, $\big\langle(2\sigma(z,a))^n\big\rangle_\rho$ is called \emph{$n$-th annihilation moment} of $\rho$.
\end{definition}

\begin{theorem}\label{g-finite-moments} 
The $n$-th moment of an observable $A$ in the state $\rho$ holds
\begin{gather}\label{eq:n-moment-derivate}
\langle A^n\rangle_\rho=(-i)^n{\frac{d^n}{dt^n}\tr{\rho e^{itA}}}\rE{t=0}\,, \quad n\in\mathbb N\,.
\end{gather}
\end{theorem}
\begin{proof}
We first consider $\Lambda=iA$ and $t\in[0,T]$, with $T>0$. So for $\epsilon>0$, one computes at once by \eqref{eq:Yosida-resolvent-onH} that the contractions $ e^{\pm t\Lambda_\epsilon}$ converge strongly to the contractions $e^{\pm t\Lambda}$ (respectively) as $\epsilon\to 0$,  uniformly for $t$. Denote
\footnotesize
\begin{align*}
f_\epsilon(t)=\frac{n!}{2}\tr{\rho\lrp{\frac{e^{t\Lambda_\epsilon}-2I+(-1)^n e^{-t\Lambda_\epsilon}}{t^n}}}\quad;\quad
f(t)&=\frac{n!}{2}\tr{\rho \lrp{\frac{e^{t\Lambda}-2I+(-1)^n e^{-t\Lambda}}{t^n}}}\,.
\end{align*}\normalsize
Note that $E=\{e^{\pm t\Lambda_\epsilon}, e^{\pm t\Lambda}: \epsilon \geq 0, t\in [0,T] \}$ is bounded. Thence, the weak and $\sigma$-weak (or weak*) operator topologies coincide on $E$. Since strong convergence implies weak convergence, then $\displaystyle \tr{\rho (e^{\pm t\Lambda_\epsilon}-e^{\pm t\Lambda})}\to 0$ uniformly for $t\in[0,T]$ as $\epsilon\to 0$. Hence $f_\epsilon \to f$ uniformly on $[0,T]$ as $\epsilon\to 0$. Besides, 
\begin{gather}\label{eq:n-derivada-Lambda}
\lim_{t\to 0} f_\epsilon(t)={\frac{d^n}{dt^n}\tr{\rho e^{t\Lambda_\epsilon}}}\rE{t=0} =\tr{\rho (\Lambda_\epsilon)^n}\,,
\end{gather} 
due to norm-derivatives imply weak*-derivatives. In this fashion, the limits can be interchange (cf.\cite[Thm.\,7.11]{MR0385023}) and by \eqref{eq:n-derivada-Lambda}, 
\begin{align*}
\lim_{\epsilon\to0}\tr{\rho (\Lambda_\epsilon)^n}=\lim_{t\to0}\lim_{\epsilon\to0}f_\epsilon(t)={\frac{d^n}{dt^n}\tr{\rho e^{t\Lambda}}}\rE{t=0}\,,
\end{align*}
which completes the proof.
 \end{proof}

For a Gaussian state $\rho=\rho(w,S)$, it is a simple matter to verify from \eqref{eq:symplectic-exponential-representation} that, for $\alpha\in\mathbb R$,
\begin{gather}\label{eq:Wtz-gaussian-representation}
e^{i\alpha(w,z)-\frac12 \alpha^2(z,Sz)}=\sqrt{\pi}\qF\rho(-\alpha z)=\tr{\rho e^{2i\alpha\sigma(z,a)}}\,.
\end{gather}
So, the following assertion is straightforward by virtue of \eqref{eq:n-moment-derivate} and \eqref{eq:Wtz-gaussian-representation}. 
\begin{corollary}
If $\rho(\omega,S)$ is a Gaussian state and $n\in\mathbb N$, then 
\begin{gather*}
\big\langle(2\sigma(z,a))^n\big\rangle_\rho=(-i)^n{\frac{d^n}{dt^n}e^{it(\omega,z)-\frac{1}{2}t^2(z,Sz)}}\rE{t=0}\,.
\end{gather*} Consequently, all annihilation moments of $\rho(w, S)$ are finite.
\end{corollary}

The characteristic function of the annihilation observable $\sigma(z,a)$ in a state $\rho$ is
\begin{align*}
\varphi_z[\rho](t)\ceq\mathcal F[\rho](tz)=\frac{1}{\sqrt{\pi}}\tr{\rho W_{tz}}\,, \quad t\in\mathbb R\,.\end{align*}
In particular,  $\varphi_z[\rho](t)=(\pi)^{-1/2}e^{-it(\omega,z)-\frac{1}{2}t^2 (z,Sz)}$, when $\rho=\rho(\omega,S)$ is Gaussian.
\begin{corollary}\label{necessity} If $\rho=\rho(\omega,S)$ is a Gaussian state, then the sequence of annihilation moments $\lrp{\big\langle \sigma(z, a)^n\big\rangle_\rho}_{n\in\mathbb N}$ satisfies the following recurrence relation  
\begin{align}\label{g-moments}
\big\langle \big(2\sigma(z,a)-(w,z)I\big)^n\big\rangle_\rho =\begin{cases}
0\,,&\mbox{for $n$ odd}\\
(z,Sz)^\frac{n}{2}(n-1)!!\,,&\mbox{for $n$ even}
\end{cases}
\end{align}
\end{corollary}
\begin{proof} It is straightforward from \eqref{eq:Wtz-gaussian-representation}  that
\begin{gather*}
e^{-\frac12 t^2(z,Sz)}=\tr{\rho e^{it(2\sigma(z,a)-(w,z)I)}}\,.
\end{gather*} In this fashion, one obtains by Theorem \ref{g-finite-moments} that
\begin{align*}
\big\langle \big(2\sigma(z,a)-(w,z)I\big)^n\big\rangle_\rho&= (-i)^{n}\frac{d^n}{dt^n}\tr{\rho e^{it(2\sigma(z,a)-(w,z)I)}}\rE{t=0}\\&=(-i)^{n}\frac{d^n}{dt^n}{e^{-\frac12 t^2(z,Sz)}}\rE{t=0}\,,
\end{align*}
whence one arrives at \eqref{g-moments}.
\end{proof}
\begin{remark}
Since $\sigma(z,a)$ and $(w,z)I$ commute, the exponential  $e^{it(2\sigma(z,a)-(w,z)I)}$ can be disentangled to obtain after taking derivatives that
\begin{gather*}
\big\langle \big(2\sigma(z,a)-(w,z)I\big)^n\big\rangle_\rho= \sum_{k=0}^{n}(-1)^k \begin{pmatrix} n \\ k\end{pmatrix}\big\langle \big(2\sigma(z,a)\big)^{n-k}\big\rangle_\rho (w,z)^k\,.
\end{gather*}
Given the first two moments $m_1(z), m_2(z)$, the implicit recurrence relation in \eqref{g-moments} is explicitly written for the moments $m_n(z)=\langle(2\sigma(z,a))^n\rangle_\rho$, $n\geq 3$,  as follows 
\begin{align*}
m_n(z)=\,\unit_{2{\mathbb N}}(n)(m_2(z)-m_1^2(z))^{\frac{n}{2}}(n-1)!!+\sum_{k=1}^n(-1)^{k+1}\begin{pmatrix} n \\ k\end{pmatrix}m_{n-k}(z)m_1^k(z)\,, 
\end{align*}
where $\unit_{2{\mathbb N}}$ denotes the indicator function of the set of even natural numbers.
\end{remark}
We will show in the next section that the recurrence relation \eqref{g-moments} allows us to characterize the Gaussian states.

\section{A characterization of Gaussian states}\label{characterization-gaussian}
Under suitable assumptions, in this section we prove a converse of Corollary \ref{necessity}. We start with a rigorous definition of the expected value for unbounded observables and a characterization of the covariance matrix $S$ of a quantum Gaussian state.

For any state $\rho$ and any bounded positive selfadjoint operator $A$, we have that 
\begin{gather*}
\tr{\rho A} = \tr{\rho^{\frac{1}{2}} A \rho^{\frac{1}{2}}} =\ip{A^{\frac{1}{2}}\rho^{\frac{1}{2}}}{ A^{\frac{1}{2}} \rho^{\frac{1}{2}}}_2 =\ip{ \rho^{\frac{1}{2}}}{A \rho^{\frac{1}{2}}}_2 \,.
\end{gather*}
In this fashion, one can represent states by unit vectors $\rho^{\frac{1}{2}}$ in the Hilbert space $L_{2}(\initsp)$. Moreover, one can carry the observables and semigroups defined on $\initsp$ to corresponding observables and semigroups acting on the space $L_2(\initsp)$. Indeed, following \cite{MR363279}, we consider the more general multiplication operator $M_B$ defined by 
\begin{gather*}
M_B \rho = B\rho\,,\quad B\in{\mathcal B}(\initsp)
\end{gather*}
which is  bounded  on $L_{2}(\initsp)$, with norm $\no{M_B}=\no{B}$. Regard the isomorphism ${\mathcal V}$ from $L_2(\initsp)$ onto $\initsp\otimes \initsp$ defined by 
\begin{gather*}
\mathcal V\vk u\vb v=u \otimes \theta v\,,
\end{gather*}
which is extended by linearity and continuity to the whole space $\initsp$, where $\theta$ is any anti-unitary operator on $\initsp$ such that $\theta^2 =\unit$. One directly computes that
\begin{gather*}
\mathcal VM_B\mathcal V^{-1}u \otimes v= \mathcal V\vk{B u}\vb{\theta v} = \lrp{B\otimes \unit}u\otimes v\,.
\end{gather*} 
Thereby, it follows by linearity and density that 
\begin{gather*}
\mathcal VM_B\mathcal V^{-1}= B\otimes \unit\,.
\end{gather*}
Hence, we can identify $L_{2}(\initsp)$ with $\initsp \otimes\initsp$ and consider the operator $B\mapsto B\otimes \unit$ instead of $M_B$. 

If $U(t)$ is a strongly continuous unitary group on $\initsp$ and $A$ is the corresponding unbounded selfadjoint generator, with associated spectral measure $(E_{\lambda} )_{\lambda\in{\mathbb R}}$. Then 
\begin{gather*}
U(t)= \int e^{it\lambda} d E_{\lambda}\,.
\end{gather*}
The unitary group $U(t)\otimes \unit$ and the spectral measure $(E_{\lambda}\otimes \unit)$ have the corresponding properties, in particular, 
\begin{gather*}
U(t)\otimes \unit= \int e^{it\lambda} d (E_{\lambda}\otimes \unit) \,.
\end{gather*}

Let $\lrp{{\mathbb U}(t)}_{t\in{\mathbb R}}$ and $\lrp{{\mathbb E}_{\lambda}}_{\lambda\in{\mathbb R}}$ be the corresponding unitary group and spectral family on $L_2(\initsp)$, such that
\begin{gather*}
{\mathcal V}{\mathbb U}_{t}{\mathcal V}^{-1} = U(t)\otimes \unit\,,\quad\forall\, t\in{\mathbb R}\,.
\end{gather*}
Consider ${\mathbb A}$ the corresponding selfadjoint generator and the representations
\begin{gather*}
{\mathbb A}=\int \lambda d{\mathbb E}_{\lambda}\quad;\quad {\mathbb U}_{t}=\int e^{it\lambda} d{\mathbb E}_\lambda\,,
\end{gather*} 
whence if $A$ is positive, then so is ${\mathbb A}$.
From \cite[Lem.\,2]{MR363279}, the explicit action of ${\mathbb A}$ is  
\begin{gather}\label{explicit-action}
{\mathbb A}\eta = \Lambda \eta, \qquad\textrm{for all } \, 
\eta \in \dom({\mathbb A})=\{\eta\in L_{2}(\initsp) \, : \, \Lambda \eta\in L_{2}(\initsp)\}\,.
\end{gather}
\begin{remark}\label{rm:Yosida-positive-selfadjoint}
For a positive selfadjoint operator $A$ and $\epsilon>0$, one has that $(I+\epsilon A)^{-1}$ is a positive contraction in $\mathcal B(\initsp)$ (cf. \cite[Sec.\,3]{MR0226433}). Thereby, $A$ is the infinitesimal generator of a strongly continuous semigroup of contractions $e^{-tA}$ and satisfies \eqref{eq:Yosida-resolvent-onH}-\eqref{eq:Yosida-resolvent-dense}. So, $A_\epsilon= A(I + \epsilon A)^{-1}$ is also positive and due to the resolvent commutes with $A$, the operators $(I+\epsilon A)^{-1}, A_\zeta, A_\eta$ pairwise commute, for $\epsilon,\zeta, \eta>0$.
\end{remark}

\begin{lemma}\label{lem:increase-positive-Yosida}
For $\epsilon>0$, if $A$ is a positive selfadjoint operator, then $(I+\epsilon A)^{-1}$ and $A_\epsilon$ increase respectively to $I$ and $A$, strongly as $\epsilon\to 0$.
\end{lemma}
\begin{proof}
If $\eta>\epsilon>0$, then the second resolvent identity implies that 
\begin{align*}
(I+\epsilon A)^{-1}-(I+\eta A)^{-1}&=(I+\epsilon A)^{-1}(\eta A-\epsilon A)(I+\eta A)^{-1}\\&=(\eta-\epsilon)A_\epsilon(I+\eta A)^{-1}\geq0\,.
\end{align*}
This also involves that  $A_\epsilon-A_\eta=(\eta-\epsilon)A_\eta A_\epsilon\geq0$. Hence by Remark~\ref{rm:Yosida-positive-selfadjoint}, the assertion follows. 
\end{proof}

The Spectral Theorem allows one to decompose any normal operator as a linear combination of two selfadjoint operators (see for instance \cite[Ch.\,3]{MR919948}). Indeed, 
\begin{itemize}
\item If $\Lambda = \int \lambda d E_{\lambda}$ is any normal operator, then by setting 
\begin{gather*}
\re\Lambda\ceq\int (\re\lambda) dE_{\lambda }\quad  \mbox{;}\quad \im\Lambda\ceq\int (\im\lambda) d E_{\lambda}\,,
\end{gather*}
we obtain two selfadjoint operators having commuting spectral projections such that  $\Lambda= \re\Lambda + i \im\Lambda$. 
\item If moreover $\Lambda = \int \lambda d E_{\lambda}$ is selfadjoint, then by setting 
\begin{gather*}
\Lambda_{+}\ceq\int \lambda\unit_{[0,\infty)}(\lambda) dE_{\lambda }\quad  \mbox{;}\quad \Lambda_{-}\ceq-\int \lambda \unit_{(-\infty, 0)}(\lambda) d E_{\lambda}\,,
\end{gather*}
we get two positive operators such that $\Lambda= \Lambda_{+} - \Lambda_{-}$.
\end{itemize}

\begin{definition}\label{def:A-traceable-rho}
A state $\rho$ is \emph{traceable} for a selfadjoint operator $A$ ($A$-traceable for short)  
if $\lim_{\epsilon\to 0}\tr{\rho[(A_+)_\epsilon-(A_-)_\epsilon]}$ exists. In such a case we write  
\begin{gather}\label{eq:def-rhoA-traceable}
 \tr{\rho A}\ceq\lim_{\epsilon\to 0}\tr{\rho[(A_+)_\epsilon-(A_-)_\epsilon]}\,.
\end{gather}
\end{definition}
\begin{remark}\label{rm:traceable-parts-A}  For a state $\rho$ and a selfadjoint operator $A=A_+-A_-$, if $\rho$ is 
traceable for both $A_+$ and $A_-$, then so is  for $A$ and 
\begin{gather*}
\tr{\rho A}=\tr{\rho A_+}-\tr{\rho A_-}\,.
\end{gather*}
However, the converse is not true. Actually, for an unbounded, positive, selfadjoint operator $B$, for which $\rho$ is not traceable, one has that $A=B-B$ is an unbounded selfadjoint operator. Thus, since $B_\epsilon\in\mathcal B(\initsp)$, 
\begin{gather*}
\tr{\rho A}=\lim_{\epsilon\to0}\tr{\rho(B_\epsilon-B_\epsilon)}=0\,,
\end{gather*}
i.e, $\rho$ is $A$-traceable. 
\end{remark}

Dentition~\ref{def:A-traceable-rho} can be extended for any normal operator $A$ as follows  
\begin{align*}
\tr{\rho A}\ceq\lim_{\epsilon\to 0}\tr{\rho\big([(\re A)_+)]_\epsilon-[(\re A)_-)]_\epsilon
+i[(\im A)_+)]_\epsilon-i[(\im A)_-)]_\epsilon\big)}\,,
\end{align*}
whenever the limit exists.
 
 \begin{remark}
 The notation \eqref{eq:def-rhoA-traceable} extends the usual one for bounded selfadjoint operators. Indeed, if $A$ is a bounded selfadjoint operator then so are $A_+$ and $A_-$. Thus, for any sate $\rho$, one has by Remark~\ref{rm:Yosida-positive-selfadjoint} that 
\begin{align*}
\lim_{\epsilon\to 0}\tr{\rho(A_+)_\epsilon}=\lim_{\epsilon\to 0}\ip{ \rho^{\frac{1}{2}}}{(A_+)_\epsilon \rho^{\frac{1}{2}}}_2=\ip{ \rho^{\frac{1}{2}}}{A_+ \rho^{\frac{1}{2}}}_2=\tr{\rho A_+}\,.
\end{align*}
The same holds for $A_-$ and, hence, 
\begin{gather*}
\lim_{\epsilon\to 0}\tr{\rho[(A_+)_\epsilon-(A_-)_\epsilon]}=\tr{\rho A}\,.
\end{gather*}
In this fashion, any state is traceable for any bounded selfadjoint operator.
 \end{remark}

The following result is a slightly different version of \cite[Lemma\,2.1]{MR1706597}. We will freely use spectral functional calculus for unbounded selfadjoint operators (see \cite[Sec.\,5.3]{MR2953553} for more details).
\begin{lemma}\label{lem-positive-traceable-rho}
A state $\rho$ is traceable for a positive selfadjoint operator $A$ if and only if $A^{\frac12}\rho^{\frac12}\in L_2(\initsp)$. In such a case 
\begin{gather}\label{eq:trace-rho-postive}
\tr{\rho A}=\int\lambda d\ip{\rho^{\frac12}}{\mathbb E\rho^{\frac12}}_2=\ip{A^{\frac12}\rho^{\frac12}}{A^{\frac12}\rho^{\frac12}}_2\,,
\end{gather}
where $\mathbb E$ is the spectral measure of \eqref{explicit-action}. If the state has  the spectral decomposition $\rho = \sum_{k\in\mathbb N} \rho_{k} \vk{u_k}\vb{u_k}$, then
\begin{gather}\label{eq:trace-rho-decomposition}
\tr{\rho A}=\sum_{k\in\mathbb N} \rho_k\no{A^{\frac{1}{2}}u_k}^2\,.
\end{gather}
\end{lemma}

\begin{proof}
 If $\rho$ is $A$-traceable then from Lemma~\ref{lem:increase-positive-Yosida}, for a unit vector $h\in\initsp$,
\begin{gather*}
\ip{\rho^{\frac12}h}{A_\epsilon\rho^{\frac12}h}\leq\tr{\rho^{\frac12}A_{\epsilon}\rho^{\frac12}}=\tr{\rho A_\epsilon}<\infty 
\end{gather*} for all $\epsilon>0$.
Thereby, $\ip{\rho^{\frac12}h}{A_\epsilon\rho^{\frac12}h}$ increases as $\epsilon\to0$, to the finite value $\|A^{\frac12}\rho^{\frac12}h\|^2$ and one obtains that $\ran \rho^{\frac12}\subset\dom A^{\frac12}$. Besides, if $\{u_k\}_{k\in\mathbb N}$ is an orthonormal basis, by the monotone convergence theorem, 
\begin{align}\label{eq:trace-rho-postive-A}
\begin{split}
\no{A^{\frac12}\rho^{\frac12}}_2^2&=\sum_{k\in\mathbb N}\no{A^{\frac12}\rho^{\frac12}u_k}^2=\lim_{\epsilon\to0}\sum_{k\in\mathbb N}\ip{\rho^{\frac12}u_k}{A_\epsilon\rho^{\frac12}u_k}=\tr{\rho A}<\infty\,.
\end{split}
\end{align}
Hence, $A^{\frac12}\rho^{\frac12}\in L_2(\initsp)$. 
Conversely, if $A^{\frac12}\rho^{\frac12}\in L_2(\initsp)$ then $\rho^{\frac12}\in\dom \mathbb A^{\frac12}$, i.e., $\int\lambda d\ip{\rho^{\frac12}}{\mathbb E\rho^{\frac12}}_2<\infty$. Note that 
\begin{gather*}
\lambda(1+\epsilon \lambda)^{-1}\leq \lambda\,,\quad \epsilon,\lambda\geq0\,.
\end{gather*}
Thus, by Lebesgue's Theorem on Dominated Convergence, 
\begin{align}\label{eq:trace-rho-postive-E}
\begin{split}
\tr{\rho A}&=\lim_{\epsilon\to0}\ip{\rho^{\frac12}}{A(I-\epsilon A)^{-1}\rho^{\frac12}}_2\\&=\lim_{\epsilon\to0}\int\frac{\lambda}{1+\epsilon\lambda}d\ip{\rho^{\frac12}}{\mathbb E\rho^{\frac12}}_2=\int\lambda d\ip{\rho^{\frac12}}{\mathbb E\rho^{\frac12}}_2\,,
\end{split}
\end{align}
as required. Equalities \eqref{eq:trace-rho-postive} and \eqref{eq:trace-rho-decomposition} are straightforward from \eqref{eq:trace-rho-postive-A} and \eqref{eq:trace-rho-postive-E}.
\end{proof}
\begin{remark}\label{rmk:linearity-traceable}
For a state $\rho$ and $A=A^*$, it is clear from \eqref{eq:def-rhoA-traceable} that if $\rho$ is $A$-traceable and $\alpha\in\mathbb R$ then $\tr{\rho(\alpha A)}=\alpha\tr{\rho A}$. Besides, if $\rho$ is traceable for both $A_\pm$, then from Remark~\ref{rm:traceable-parts-A} and \eqref{eq:trace-rho-postive}, $\rho$ is $A$-traceable and 
\begin{align*}
\tr{\rho A}=\int\lambda_+ d\ip{\rho^{\frac12}}{\mathbb E\rho^{\frac12}}_2-\int\lambda_- d\ip{\rho^{\frac12}}{\mathbb E\rho^{\frac12}}_2=\int\lambda d\ip{\rho^{\frac12}}{\mathbb E\rho^{\frac12}}_2.
\end{align*}
Moreover, for $\alpha,\beta\in\mathbb R$ and $\mathbb E$-a.e. finite real-valued Borel functions $f,g$ such that $\rho$ is traceable for both $f(A)_\pm$ and $g(A)_\pm$, one has that  
\begin{align*}
\tr{\rho(\alpha f(A)+\beta g(A))}&=\int\lrp{\alpha f(\lambda)+\beta g(\lambda)} d\ip{\rho^{\frac12}}{\mathbb E\rho^{\frac12}}_2\\
&=\,\alpha\int f(\lambda)d\ip{\rho^{\frac12}}{\mathbb E\rho^{\frac12}}_2+\beta\int g(\lambda)d\ip{\rho^{\frac12}}{\mathbb E\rho^{\frac12}}_2\\&=\,\alpha\tr{\rho f(A)}+\beta\tr{\rho g(A)}\,.
\end{align*}
\end{remark}

\begin{theorem}\label{thm:An-traceable-derivates}
Let $\rho$ be a state and $A=A^*$. If $\rho$ is traceable for both $(A^n)_\pm$, then
\begin{gather}\label{eq:An-traceable-derivates}
\tr{\rho A^n}=i^{-n}\frac{d^n}{dt^n}\tr{\rho e^{itA}}\rE{t=0}\,,\quad n\geq 0\,.
\end{gather}
Hence, Theorem~\ref{g-finite-moments} implies $\tr{\rho A^n}=\langle A^n\rangle_\rho$.
\end{theorem}
\begin{proof}
One has by virtue of Lemma~\ref{lem-positive-traceable-rho}  that 
\begin{align*}
\int\abs{\lambda^n} d\ip{\rho^{\frac12}}{\mathbb E\rho^{\frac12}}_2\leq\int(\lambda^n)_+ d\ip{\rho^{\frac12}}{\mathbb E\rho^{\frac12}}_2+\int(\lambda^n)_- d\ip{\rho^{\frac12}}{\mathbb E\rho^{\frac12}}_2<\infty\,.
\end{align*}
Now the estimate, 
\begin{gather*}
\abs{\frac{e^{i h \lambda} - 2 + (-1)^ne^{-i h \lambda}}{h^{n}}} \leq 2 \abs{\lambda}^n 
\end{gather*} and Lebesgue's Theorem on Dominated Convergence allows interchanging the limit with the integral in the expression 
\begin{align*}
\frac{d^n}{dt^n}\tr{\rho e^{itA}}\rE{t=0}&=\lim_{h\to0}\frac{n!}{2}\ip{\rho^{\frac12}}{ \lrp{\frac{e^{ihA}-2I+(-1)^n e^{-ihA}}{h^n}}\rho^{\frac12}}_2\\
&=\lim_{h\to0}\int\frac{n!}{2}\lrp{\frac{e^{i h \lambda} - 2 + (-1)^ne^{-i h \lambda}}{h^{n}}}d\ip{\rho^{\frac{1}{2}}}{{\mathbb E}_{\lambda} \rho^{\frac{1}{2}}}_2\\&=\int i^n\lambda^nd\ip{\rho^{\frac{1}{2}}}{{\mathbb E}_{\lambda} \rho^{\frac{1}{2}}}_2=i^n\tr{\rho A^n}\,,
\end{align*}
which completes the proof.
\end{proof}

\begin{remark} If a Gaussian state $\rho$ is traceable for both $\lrp{ \sigma(z, a)^n}_\pm$, then from Theorem~\ref{thm:An-traceable-derivates}, the $n$-th annihilation moment holds
\begin{gather*}
\big\langle(2 \sigma(z,a))^n\big\rangle_\rho =\tr{\rho(2\sigma(z, a))^n}\,,\quad n\geq0\,.
\end{gather*}
\end{remark}

\begin{corollary}\label{mean-covariance} 
Let  $\rho=\rho(w, S)$ be a Gaussian state, with  $w=\sqrt{2}(l-im)$, and consider 
$z=x + iy\in{\mathbb C}$.
\begin{enumerate}[(i)]
\item\label{it-Mom1} If $\rho$ is both $\sigma(z,a)_\pm$-traceable, then due to the recurrence formula \eqref{g-moments}, $\tr{\rho\sigma(z,a)}=(w,z)/2$.  In particular, one has for $z=\sqrt{2}$ (resp. $z=i\sqrt{2}$) that 
\begin{gather*}
\tr{\rho P}=l\quad\mbox{and}\quad \tr{\rho Q}=m\,.
\end{gather*}
\item\label{it-Mom2} If $\rho$ is $\sigma(z,a)^2$-traceable, then it follows that $\tr{\rho\sigma(z,a)^2}=\big((w,z)^2+(z,Sz)\big)/4$. Particularly,
\begin{enumerate}
\item from $z=\sqrt{2}$ that $(1, S1)=2\lrp{\tr{\rho P^2}-l^2}$,  
\item from $z= \sqrt{2}i$ that $(i, Si)=2\lrp{\tr{\rho Q^2}-m^2}$, 
\item from $z=\sqrt{2}(1+i)$ that $\lrp{1+i, S(1+i)}= 2\lrp{\tr{\rho (P-Q)^{2}} -(l-m)^{2}}$.
\end{enumerate}
Hence, 
\begin{gather*}
(1,Si)=\tr{\rho (P-Q)^2} -\tr{\rho P^2}-\tr{\rho Q^2}+2lm\,.
\end{gather*}
\end{enumerate}
\end{corollary}
\begin{proof} To prove \eqref{it-Mom2}, from Theorem~\ref{thm:An-traceable-derivates} and \eqref{eq:Wtz-gaussian-representation}, replacing $\alpha=t/2$, one obtains that 
\begin{align*}
\tr{\rho\sigma(z,a)^2}=-\frac{d^2}{dt^2}{e^{\frac12it(w,z)-\frac18 t^2(z,Sz)}}\rE{t=0}=\frac14\lrp{(w,z)^2+(z,Sz)}\,,
\end{align*}
as required.
\end{proof}

\subsection{Covariance matrix: the role of the uncertainty principle}
\begin{definition}
We call a state $\rho$ \emph{amenable} if for $n=1,2$ and all $z\in\C$, $\rho$ is traceable for both $\lrp{\sigma(z,a)^n}_\pm$.
\end{definition}
The following definition is reminiscent of the analogous definitions in classical probability for random variables. 
\begin{definition}
The \emph{variance} $V_\rho(A)^2$ of an observable $A$ in the state $\rho$ is  
\begin{gather*}
V_{\rho}(A)^2 \ceq\tr{\rho A^2} - \tr{\rho A}^2\,,
\end{gather*}
whenever the right hand makes sense. 
\end{definition}
It is clear that if $\rho$ is amenable then $V_{\rho}(P)^2$ and $V_{\rho}(Q)^2$ are well-defined. 
\begin{lemma}
Let $\rho$ be an amenable Gaussian state. Then for $z\in\C$,
\begin{gather*}
\no{\Big(\sigma(z,a)-\tr{\rho \sigma(z,a)}I\Big) \rho^{\frac{1}{2}}}_{2}^{2} = V_{\rho}\big(\sigma(z,a)\big)^2\,.
\end{gather*}
In particular, for $z=\sqrt{2}$ (\textrm{resp.} $z=-i\sqrt{2}$),
\begin{align}\label{norm-variance}
\begin{split}
\no{\big(P-\tr{\rho P}I\big) \rho^{\frac{1}{2}}}_{2}^{2} = V_{\rho}(P)^2 \qquad
\big(\textrm{resp.} \quad \no{\big(Q-\tr{\rho Q}I\big) \rho^{\frac{1}{2}}}_{2}^{2} = V_{\rho}(Q)^2\big)\,.
\end{split}
\end{align}
\end{lemma}
\begin{proof}
Since the pair $(\rho, \sigma(z,a))$ holds the assumptions of Theorem~\ref{thm:An-traceable-derivates}, if $\mathbb E$ is the spectral measure of $\sigma(z,a)$, then
\begin{align*}
V_{\rho}(\sigma(z,a))^2&=\int\lambda^2d\ip{\rho^{\frac12}}{\mathbb E\rho^{\frac12}}_2-\tr{\rho(\sigma(z,a)}^2\\&=\int\lrp{\lambda-\tr{\rho(\sigma(z,a)}}^2 d\ip{\rho^{\frac12}}{\mathbb E\rho^{\frac12}}_2\\&=\no{(\sigma(z,a)-\tr{\rho \sigma(z,a)}I) \rho^{\frac{1}{2}}}_2^2\,,
\end{align*}
as required. 
\end{proof}
\begin{definition} Let $\rho$ be a state traceable for both $P^2$ and $Q^2$ (in particular when $\rho$ is amenable). We say that the \emph{CCR holds true weakly in the Hilbert-Schmidt sense in $\rho$} if,  
\begin{gather}\label{eq:weakly-CCR-rho}
\ip{P\rho^{\frac{1}{2}}}{Q\rho^{\frac{1}{2}}}_2 -\ip{Q \rho^{\frac{1}{2}}}{P\rho^{\frac{1}{2}}}_2=-i\,,
\end{gather}
i.e., $\im \ip{P\rho^{\frac{1}{2}}}{Q\rho^{\frac{1}{2}}}_2=-1/2$.
\end{definition} 
As a consequence of Lemma~\ref{lem-positive-traceable-rho}, both operators $P\rho^{\frac{1}{2}}$ and $Q\rho^{\frac{1}{2}}$ are Hilbert-Schmidt and, hence, the inner products in \eqref{eq:weakly-CCR-rho} are finite. So, we have the following inequality, which in turn yields the uncertainty principle.
\begin{theorem}\label{th-uncertainty-principle} 
If $\rho$ an amenable Gaussian state such that  the CCR holds weakly in the Hilbert-Schmidt sense in $\rho$, then 
\begin{align}\label{uncertainty-principle}   
\begin{split}
 V_{\rho}(P)^2  V_{\rho}(Q)^2\geq \frac{1}{4}+\lrp{\tr{\rho P}\tr{\rho Q}-\re \ip{P\rho^{\frac{1}{2}}}{Q\rho^{\frac{1}{2}}}_2}^2\,.
\end{split}
\end{align} 
Hence, the uncertainty principle $V_{\rho}(P) V_{\rho}(Q) \geq 1/2$ is true.
 \end{theorem}
\begin{proof} 
Due to our assumptions, operators $P\rho^{\frac{1}{2}}, Q\rho^{\frac{1}{2}}$ and $\rho^{\frac{1}{2}}$, belong to $L_2(\initsp)$. Besides, for $\hat P=P-\tr{\rho P}I$ and $\hat Q=Q-\tr{\rho Q}I$, one checks at once by \eqref{eq:weakly-CCR-rho} that
\begin{align*}
\ip{\hat P\rho^{\frac{1}{2}}}{\hat Q\rho^{\frac{1}{2}}}_2&=\ip{P\rho^{\frac{1}{2}}}{Q\rho^{\frac{1}{2}}}_2 - \tr{\rho P}\tr{\rho Q}\\&=
\re \ip{P\rho^{\frac{1}{2}}}{Q\rho^{\frac{1}{2}}}_2 -\tr{\rho P}\tr{\rho Q}-\frac i2\,.
\end{align*}
Therefore, by Schwarz inequality, 
\begin{align*}
\no{\hat P\rho^{\frac{1}{2}}}_2^2\no{\hat Q\rho^{\frac{1}{2}}}_2^2
&\geq\abs{\ip{(P-\tr{\rho P}I)\rho^{\frac{1}{2}}}{(Q-\tr{\rho Q}I)\rho^{\frac{1}{2}}}_2}^2\\&=\lrp{\tr{\rho P}\tr{\rho Q}-\re \ip{P\rho^{\frac{1}{2}}}{Q\rho^{\frac{1}{2}}}_2}^2+\frac14\,,
\end{align*}
wherefrom by \eqref{norm-variance}, one arrives at \eqref{uncertainty-principle}.
\end{proof}

We now describe the covariance matrix of an amenable Gaussian state. 

\begin{proposition}\label{prop:entries-positiveS} If $\rho=\rho(w,S)$ is an amenable Gaussian state, then $S$ is positive definite, with entries
\begin{align}\label{eq:entries-S-amenable}
\begin{split}
(1,S1)&=2V_\rho(P)^2\,,\qquad(i,Si)=2V_\rho(Q)^2\,,\\
(1,Si)&=2\lrp{\tr{\rho P}\tr{\rho Q}-\re \ip{P\rho^{\frac{1}{2}}}{Q\rho^{\frac{1}{2}}}_2}\,.
\end{split}
\end{align}
\end{proposition}
\begin{proof}
It is clear from Remark~\ref{rmk:linearity-traceable} that $\rho$ is traceable for $(2\sigma(z,a)-(w,z)I)^2$, since it is amenable. Thus, Lemma~\ref{lem-positive-traceable-rho}, Theorem~\ref{thm:An-traceable-derivates} and \eqref{g-moments} imply
\begin{align*}
(z,Sz)&=\langle (2\sigma(z,a)-(w,z)I)^2\rangle_{\rho}\\&=\tr{\rho(2\sigma(z,a)-(w,z)I)^2}\geq0\,.
\end{align*}
Now, note that $\rho$ is also traceable for $P^2,Q^2$ and $(P-Q)^2$. So, Lemma~\ref{lem-positive-traceable-rho} fulfils  
\begin{align*}
\tr{\rho(P-Q)^2}&=\no{(P-Q)\rho^{\frac12}}_2^2\\&=\tr{\rho P^2}+\tr{\rho Q^2}-2\re \ip{P\rho^{\frac{1}{2}}}{Q\rho^{\frac{1}{2}}}_2\,,
\end{align*}
whence from Corollary~\ref{mean-covariance}, one produces \eqref{eq:entries-S-amenable}. 
\end{proof}
The following is straightforward from Theorem~\ref{th-uncertainty-principle} and Proposition~\ref{prop:entries-positiveS}.

\begin{corollary}\label{positivity-S}
If $\rho$ is a Gaussian state with covariance matrix $S$ and satisfies the assumptions in Theorem~\ref{th-uncertainty-principle}, then 
\begin{gather}\label{uncertainty-p-1}
(1,S1)(i,Si)\geq (1,Si)^2 + 1 \,.
\end{gather}
Therefore, $S$ is an invertible matrix such that 
\begin{gather}\label{uncertainty-p-2}
S-i\Sigma\geq 0\,,\quad\mbox{where}\quad \Sigma=\begin{pmatrix}
0&-1\\1&0
\end{pmatrix}\,.
\end{gather} 
 \end{corollary}
\begin{remark} Some significant notes related to this section:
\begin{itemize}
\item[(i)] We emphasize that we work with the adjugate (or classical adjoint) $S$ of the conventional covariance matrix usually used in many other good references on Gaussian states. This explains the minus sign of the off-diagonal matrix elements of $S$ in  \eqref{eq:entries-S-amenable}.
\item[(ii)] A different proof of Corollary \ref{positivity-S} and its converse, based on the Quantum Bochner's theorem in $L_{2}({\mathbb R})$ and properties of infinitely divisible positive definite kernels, can be found in \cite[Thm.\,3.1]{MR2662722}.  
\item[(iii)] K.R. Parthasarthy interpreted inequality \eqref{uncertainty-p-2} as the \textit{complete uncertainty principle} (c.f. \cite[Rmk.\,3.2]{MR2662722}). We emphasize that \eqref{uncertainty-p-1} and, hence, \eqref{uncertainty-p-2} correspond with Schwartz inequality \eqref{uncertainty-principle}. 
\end{itemize}
\end{remark}

\subsection{A characterization of Gaussian states}
We are ready to prove the converse of Corollary \ref{necessity}, which completes a characterization of quantum Gaussian states in terms of its annihilation moments.  
\begin{definition}
The \emph{moment-generating} function of an observable, i.e., a selfadjoint operator $A$ in the state $\rho$, is given by 
\begin{gather*}
g_{\rho,A}(t)=\sum_{n\geq 0} \frac{1}{n!}\langle A^n\rangle_{\rho} t^n\,,
\end{gather*}
whenever all moments are finite and the series converges.
\end{definition}
\begin{lemma}\label{generating-moments} 
Let $\rho$ be a traceable state for both $(\sigma(z,a)^n)_\pm$, for all $n\geq 0, z\in\C$ and suppose that the sequence of annihilation moments in $\rho$ satisfies the recurrence relation 
\eqref{g-moments}, for some $w\in{\mathbb C}$ and symmetric real matrix $S$. Then, the moment-generating function of the observable $2\sigma(z,a)$ in $\rho$ fulfils 
\begin{gather*}
g_{\rho,2\sigma(z,a)}(t)=e^{(w,z)t + \frac{1}{2}(z, Sz)t^2}\,,\quad  t\in{\mathbb R}\,.
\end{gather*}
\end{lemma}
\begin{proof}
Let $\mathbb E$ be the spectral measure of $\sigma(z,a)$ on $L_2(\initsp)$. By Theorem~\ref{thm:An-traceable-derivates}, Lebesgue's Theorem on Dominated Convergence and \eqref{g-moments}, one has  that 
\begin{align}\label{proof-lemma-generating}
\begin{split}
e^{-(w,z)t}\sum_{n\geq 0} \frac{1}{n!}\tr{\rho\big(2\sigma(z,a)\big)^n} t^n = & \int e^{-t(w,z)+2t\lambda} d\ip{\rho^{\frac12}}{\mathbb E\rho^{\frac12}}_2\\=&\sum_{n\geq 0} \frac{1}{n!}\tr{\rho \big(2\sigma(z,a)- (w,z)I\big)^n} t^n \\=&\sum_{n\geq 0} \frac{1}{n!}\big\langle\big(2\sigma(z,a)- (w,z)I\big)^n\big\rangle_{\rho} t^n \\=&\sum_{n\geq 0} \frac{1}{(2n)!} (z,Sz)^n (2n-1)!! t^{2n} \\=& \sum_{n\geq 0} \frac{1}{n!}\big(\frac{(z,Sz)}{2}\big)^n t^{2n}=e^{\frac{1}{2}(z,Sz)t^2}\,,
\end{split}
\end{align}
since $(2n-1)!!=(2n)!(2^nn!)^{-1}$.
\end{proof}

\begin{theorem}\label{characterization-gaussian-state} 
If $\rho$ is a state such that the pair $(\rho, \sigma(z,a)^n)$ satisfies conditions of the previous Lemma \ref{generating-moments}, then $\rho$ is Gaussian. 
\end{theorem}

\begin{proof} 
Replacing $t$ by $-i$ in \eqref{proof-lemma-generating}, we get \[\tr{\rho W_z}=\sum_{n\geq 0} \frac{1}{n!}\tr{\rho\big(-2i\sigma(z,a)\big)^n} =e^{-i(w,z)-\frac12(z,Sz)}\] Consequently, $\rho$ is Gaussian. 
\end{proof}

\section{Gaussian channels and semigroups}\label{channels-semigroups}
It is well-known that a channel is a completely positive trace preserving map acting on ${\mathcal B}(\initsp)$. Besides, if $T$ is such a channel in Schr$\ddot{\rm o}$dinger (or predual) representation and $\rho_{in}=\rho$ is an initial state, then $T$ maps $\rho$ in an output state $\rho_{out}=T(\rho)$. For $\initsp=L_{2}({\mathbb R})$, we say that $T$ is a \textit{Gaussian channel} if for any initial Gaussian state $\rho$ the output state $T(\rho)$ is also Gaussian.
\begin{example}[Coherent channel]
For $u\in\C$, the coherent channel $T_u(\rho)=W_u\rho W_u^*$ is Gaussian. Indeed, if $\rho=\rho(w,S)$ is a Gaussian state, then 
\begin{gather*}
\sqrt{\pi}\qF{W_u\rho W_u^*}(z)=e^{-2i\sigma(u,z)}\qF\rho(z)=e^{-i(w-2iu,z)-\frac12(z,Sz)}\,,
\end{gather*} 
i.e., $T_u(\rho)$ is a Gaussian state with mean value $w-2iu$ and covariance $S$. In particular, if $\rho$ is the vacuum state (v.s. Example~\ref{ex:vacuum-satate}), then the coherent state $\vk{\psi_u}\vb{\psi_u}=T_u(\vk{\unit}\vb{\unit})$ is Gaussian with mean value vector $-2iu$ and covariance matrix $I$.
\end{example}

\begin{example}[Squeezing channel]\label{eg:Squeezing-channel} We regard the \emph{squeezing} operator given by
\begin{gather*}
S_\zeta\ceq e^{\frac12(\zeta{a^\dagger}^2-\cc \zeta a^2)}\,,\qquad \zeta\in\C
\end{gather*}
which is unitary and satisfies $S_\zeta^*=S_{-\zeta}$. Now, if we denote $A=-\frac12(\zeta{a^\dagger}^2-\cc \zeta a^2)$, then a simple computation shows that $\ad_A(za^\dagger-\overline{z} a)=-(\zeta\cc za^\dagger-\cc\zeta za)$, where 
$\ad_X (Y)=[X,Y]$. In this fashion, for $\zeta=re^{i\theta}$ and $k\geq0$, one produces that  
\begin{align*}
\ad^{2k}_A(za^\dagger-\overline{z} a)=r^{2k}(za^\dagger-\cc{z} a)\quad\mbox{;}\quad \ad^{2k+1}_A(za^\dagger-\cc{z}a)=(-1)^{2k+1}r^{2k}(\zeta\cc z a^\dagger-\cc{\zeta}z a)\,.
\end{align*}
So, the well-known relation $e^{A}(za^\dagger-\cc{z}a)e^{-A}=e^{\ad_A}(za^\dagger-\cc{z}a)$ implies
\begin{align*}
S_\zeta^* (za^\dagger-\cc{z} a) S_\zeta&=\sum_{k=0}^\infty \frac{1}{k!}\ad^k_A (za^\dagger-\cc{z} a) \\&=\sum_{k=0}^\infty r^{2k} \lrp{\frac{1}{(2k)!}(za^\dagger-\cc{z} a)+\frac{(-1)^{2k+1}}{(2k+1)!} (\zeta\cc{z}a^\dagger-\cc{\zeta}z a)}\\
&=(za^\dagger-\cc{z} a)\cosh r-(e^{i\theta}\cc z a^\dagger-e^{-i\theta}za)\sinh r \\
&=(z\cosh r -\cc z e^{i\theta}\sinh r)a^\dagger -(\cc{z}  \cosh r - ze^{-i\theta} \sinh r)a\\&=U_\zeta z a^\dagger- \overline{U_\zeta z }a\,.
\end{align*}
where $U_\zeta$ is the invertible, symmetric and real matrix 
\begin{gather*}
U_\zeta\ceq\begin{pmatrix}
\cosh r-\sinh r\cos\theta&-\sinh r\sin\theta\\
-\sinh r\sin\theta&\cosh r+\sinh r\cos\theta
\end{pmatrix}\in{\mathcal B}_{\mathbb R}({\C})
\end{gather*}
one readily  (c.f. \cite{MR1452194}) obtains  
\begin{gather}\label{eq:ordering-WS}
W_z S_\zeta=S_\zeta W_{U_\zeta z}\,,\qquad U_\zeta z=z\cosh r-\cc ze^{i\theta}\sinh r\,.
\end{gather}

Hence, one can assert that the squeezing channel $\mathcal S_\zeta(\rho)=S_\zeta\rho S_\zeta^*$ is a Gaussian channel. Indeed, for a Gaussian state $\rho=\rho(w,S)$, one computes by \eqref{eq:ordering-WS} that
\begin{align*}
\sqrt{\pi}\qF{S_\zeta\rho S_\zeta^*}(z)&=\tr{\rho S_\zeta^*W_zS_\zeta}=\tr{\rho W_{U_\zeta z}}\\&=e^{-i(w,U_\zeta z)-\frac12 (U_\zeta z,SU_\zeta z)}=e^{-i(U_\zeta w,z)-\frac12 (z,U_\zeta SU_\zeta z)}\,,
\end{align*}
i.e.,  $\mathcal S_\zeta(\rho)$ is Gaussian with mean value vector $U_\zeta w$ and covariance matrix $U_\zeta SU_\zeta$.
\end{example}
\begin{example}[Squeezed states] On account of \eqref{eq:ordering-WS}, one has that the order of composing the Weyl and squeezing operators is only a convention to define the squeezed states. Actually, if $\rho(w,S)$ is a Gaussian state, then the following state
\begin{gather*}
S_\zeta W_u\rho W_u^* S_\zeta^*\qquad \lrp{\mbox{resp.}\quad
W_uS_\zeta\rho S_\zeta^*W_u^*}
\end{gather*}
is also Gaussian with mean value vector $U_\zeta (w-2iu)$ and covariance matrix $U_\zeta SU_\zeta$ (resp. $U_\zeta w-2iu$ and $U_\zeta SU_\zeta$). Besides, it is easy to verify from \eqref{eq:ordering-WS} that 
\begin{gather*}
W_uS_\zeta\rho S_\zeta^*W_u^*=S_\zeta W_{U_\zeta z}\rho W_{U_\zeta z}^*S_\zeta^*\,.
\end{gather*}
\begin{remark}
The squeezed coherent state is given by
\begin{gather*}
\rho_s=S_\zeta W_u\vk{\unit}\vb{\unit} W_u^* S_\zeta^*\,,\quad(\zeta,u\in\C)
\end{gather*}
which is Gaussian with mean value vector $-2U_\zeta iu$ and covariance matrix $U_\zeta^2$. Clearly, $\det U_\zeta^2=1$ and, hence, equality holds in  Schwartz inequality  \eqref{uncertainty-p-1} (see also \eqref{uncertainty-principle}). 

Particularly, if $\zeta\in\mathbb R$ and $u=a+ib\in\C$, then the mean momentum and 
position vectors, and the covariance matrix of $\rho_s$ are given by (q.v. Remark~\ref{rm:properties-wS})
\begin{gather*}
l=\sqrt{2}e^{-\abs{\zeta}}b\,,\quad m=\sqrt{2}e^{\abs{\zeta}}a\quad\mbox{and}\quad S=\begin{pmatrix}
e^{-2\abs{\zeta}}&0\\0&e^{2\abs{\zeta}}
\end{pmatrix}\,,
\end{gather*}
respectively. Hence, the squeezed coherent states saturate the Heisenberg uncertainty principle $V_{\rho_s}(P) V_{\rho_s}(Q) = 1/2$. 
\end{remark}
\end{example}

\begin{example}[Bosonic Gaussian channels]
Let $x$ and $y$ be real-valued mean zero random variables with joint normal distribution and let $z= x + iy$. Set 
\begin{gather*}
T (\rho) = {\mathbb E} W_{z} \rho W_{z}^*\,,
\end{gather*}
where ${\mathbb E}$ denotes expectation with respect to the probability distribution of $z$. If $\rho=\rho_g (w, S)$ is a Gaussian state with ${w=\sqrt{2}(l - im)}$, then one computes that 
\begin{align*}
\qF{T(\rho)}(u)&=\frac1{\sqrt{\pi}}\E\tr{W_z\rho W_{-z}W_u}=\frac1{\sqrt{\pi}}\E e^{-2i\sigma(u,z)}\tr{\rho W_u}\\&=\E e^{-2i\sigma(u,z)}\qF\rho(u)=e^{i(\mu,u)-\frac12(u,Cu)}\lrp{\frac1{\sqrt{\pi}}e^{-i(w,u)-\frac12(u,Su)}}\\
&=\frac1{\sqrt{\pi}}e^{-i(w-\mu,u)-\frac12(u,(S+C)u)}=\qF{\rho(w-\mu,S+C)}(u)\,,
\end{align*}
where $C$ is the $2\times 2$ covariance matrix of the Gaussian random vector ${X=(-2y, 2x)}$ with mean $\mu \in \mathbb R^2$ and, thence, its characteristic function is 
\begin{gather*}
{\mathbb E}e^{-2i\sigma(u,z)}={\mathbb E}e^{i (u,X)}=e^{i(\mu,u)-\frac12(u,Cu)}\,.\end{gather*}
Therefore, $T$ is  a Gaussian channel.
\end{example}

\begin{example}[Bosonic Gaussian QMS] A quantum Markov semigroup (QMS for short) $T=(T_t)_{t\geq 0}$ is called Gaussian if for each $t\geq 0$, $T_{t}$ is a Gaussian channel.

Let $z=r + is\in\C$  and $\omega =(\omega_t)_{t\geq 0}$ be the standard classical Brownian motion, viz. for each $t\geq 0$ the random variable $\omega_t$ has a normal distribution with mean zero and variance $t$. Define 
\begin{gather*}
T_t^{z}(\rho) := {\E} W_{\omega_t z} \rho W_{\omega_t z}^* = {\E} e^{-2i\omega_t\sigma(z,a)} \rho e^{2i\omega_t\sigma(z,a)}\,,\quad (t\geq0)
\end{gather*}
viz. $T_t^z(\rho)=\E e^{-2i\omega_t\ad_{\sigma(z,a)}}\rho$ (see reasoning of Example~\ref{eg:Squeezing-channel}), where $\E$ denotes expectation w.r.t. the probability measure of the process $\omega$ in the strong Bochner sense, i.e., for each $u\in L_{2}({\mathbb R})$, the function $f(x)=W_{x z} \rho W_{x z}^* u$ is integrable in norm w.r.t. the distribution of $\omega_t$. Since the Brownian motion has independent increments, one has that $(T_t^z)_{t\geq 0}$ is a QMS \cite[E.g.\,30.1]{MR3012668}. Besides, 
\begin{align}\label{eq:Bosonic-QMS}
\begin{split}
\E e^{-2i\omega_t\ad_{\sigma(z,a)}}\rho&=\frac{1}{\sqrt{2\pi t}}\int e^{-2ix\ad_{\sigma(z,a)}}\rho e^{-\frac{x^2}{2t}}dx\\&=\frac{1}{\sqrt{2\pi}}\int e^{-2i\sqrt{t}x\ad_{\sigma(z,a)}}\rho e^{-\frac{x^2}{2}}dx\\&=\frac{1}{\sqrt{2\pi}}\int e^{-\frac12\lrp{xI+2i\sqrt{t}\ad_{\sigma(z,a)}}^2-2t\ad_{\sigma(z,a)}^2}\rho dx=e^{-2t\ad_{\sigma(z,a)}^2}\rho\,,
\end{split}
\end{align} 
i.e., $T_t^z(\rho)=e^{-2t\ad_{\sigma(z,a)}^2}\rho$. Hence, at least formally since $\sigma(z,a)$ is unbounded, the generator of $(T_t^{z})_{t\geq 0}$ is ${\mathcal L}_{z}=2\ad_{\sigma(z,a)}^2$, i.e.,  
\begin{align*}
\mathcal L_z(\rho)=2\big[\sigma(z,a),[\sigma(z,a),\rho]\big]=2\{\sigma(z,a)^2,\rho\}-4\sigma(z,a)\rho\sigma(z,a)\,.
\end{align*}
Thereby, $\mathcal L_z$ belongs to the class of the so-called \textit{quadratic super-operators} (cf. \cite{MR4065489}).
\end{example}
\begin{theorem}
For every $z=r+is$, the QMS $T^z = \big(T_{t}^z\big)_{t\geq 0}$ is Gaussian. Moreover, if $z\neq 0$, then $T^z$ has not Gaussian invariant states. 
\end{theorem} 
\begin{proof} Following the same lines of \eqref{eq:Bosonic-QMS}, one can compute for $u\in\C$ that 
\begin{gather*}
\E\lrp{e^{2i\omega_t\sigma(z,u)}}=e^{-2t(\sigma(z,u))^2}=e^{-2t(u,Cu)}\,,\quad\mbox{where}\quad
C=\begin{pmatrix}
r^2&-rs\\-rs&s^2
\end{pmatrix}\in{\mathcal B}_{\mathbb R}({\C})\,.
\end{gather*}
Thus, if $\rho=\rho(w,S)$ is a Gaussian state, then
\begin{align}\label{eq:Gaussian-channel}
\begin{split}
\sqrt{\pi}\qF{T_{t}^z (\rho)}(u)&=\E\lrp{\tr{\rho W_{-\omega_t z}W_uW_{\omega_t z}}}=\E\lrp{e^{2i\omega_t\sigma(z,u)}}\tr{\rho W_u}\\
&=e^{-2t(u,Cu)}e^{-i(w,z)-\frac12 (z,Sz)}=e^{-i(w,z)-\frac12 (z,(S+4tC)z)}\,,
\end{split}\end{align}
which proves that $T^z$ is Gaussian. Now, if $\rho$ is a Gaussian invariant state of $T^z$, then it follows by \eqref{eq:Gaussian-channel} that
\begin{gather*}
\qF{\rho}(u)=\qF{T_{t}^z (\rho)}(u)=e^{-2t\sigma(z,u)^2}\qF\rho(u)\,,
\end{gather*}
i.e., $\qF\rho=0$. This contradicts the fact that $z\neq0$ and $\mathcal F$ is injective.
\end{proof}

\begin{acknowledgments}
The financial support from CONACYT-Mexico (Grant CF-2019-684340) and postdoctoral fellowship 136135 ``Estructura de los estados estacionarios de generadores de Markov de transporte cu\'antico y del l\'imite de baja densidad'', is gratefully acknowledged. 
\end{acknowledgments}

\def\cprime{$'$} \def\lfhook#1{\setbox0=\hbox{#1}{\ooalign{\hidewidth
  \lower1.5ex\hbox{'}\hidewidth\crcr\unhbox0}}} \def\cprime{$'$}
  \def\cprime{$'$} \def\cprime{$'$} \def\cprime{$'$} \def\cprime{$'$}
  \def\cprime{$'$} \def\cprime{$'$}


\begin{thebibliography}{10}

\bibitem{MR0226433}
S.~J. Bernau.
\newblock The square root of a positive self-adjoint operator.
\newblock {\em J. Austral. Math. Soc.}, 8:17--36, 1968.

\bibitem{MR464938}
B.~Demoen, P.~Vanheuverzwijn, and A.~Verbeure.
\newblock Completely positive maps on the {CCR}-algebra.
\newblock {\em Lett. Math. Phys.}, 2(2):161--166, 1977/78.

\bibitem{MR551128}
B.~Demoen, P.~Vanheuverzwijn, and A.~Verbeure.
\newblock Completely positive quasifree maps of the {CCR}-algebra.
\newblock {\em Rep. Math. Phys.}, 15(1):27--39, 1979.

\bibitem{MR1706597}
K.~S. Garsia.
\newblock On the structure of the cone of normal unbounded completely positive
  mappings.
\newblock {\em Mat. Zametki}, 65(2):194--205, 1999.

\bibitem{MR2683408}
T.~Heinosaari, A.~S. Holevo, and M.~M. Wolf.
\newblock The semigroup structure of {G}aussian channels.
\newblock {\em Quantum Inf. Comput.}, 10(7-8):619--635, 2010.

\bibitem{MR2797301}
A.~Holevo.
\newblock {\em Probabilistic and statistical aspects of quantum theory},
  volume~1 of {\em Quaderni/Monographs}.
\newblock Edizioni della Normale, Pisa, second edition, 2011.
\newblock With a foreword from the second Russian edition by K. A. Valiev.

\bibitem{MR4218302}
T.~C. John and K.~R. Parthasarathy.
\newblock A common parametrization for finite mode {G}aussian states, their
  symmetries, and associated contractions with some applications.
\newblock {\em J. Math. Phys.}, 62(2):Paper No. 022102, 39, 2021.

\bibitem{MR919948}
P.~E.~T. Jorgensen.
\newblock {\em Operators and representation theory}, volume 147 of {\em
  North-Holland Mathematics Studies}.
\newblock North-Holland Publishing Co., Amsterdam, 1988.
\newblock Canonical models for algebras of operators arising in quantum
  mechanics, Notas de Matem\'{a}tica [Mathematical Notes], 120.

\bibitem{MR363279}
K.~Kraus and J.~Schr\"{o}ter.
\newblock Expectation values of unbounded observables.
\newblock {\em Internat. J. Theoret. Phys.}, 7:431--442, 1973.

\bibitem{MR1452194}
M.~M. Nieto and D.~R. Truax.
\newblock Holstein-{P}rimakoff/{B}ogoliubov transformations and the multiboson
  system.
\newblock {\em Fortschr. Phys.}, 45(2):145--156, 1997.

\bibitem{MR3012668}
K.~R. Parthasarathy.
\newblock {\em An introduction to quantum stochastic calculus}.
\newblock Modern Birkh\"{a}user Classics. Birkh\"{a}user/Springer Basel AG,
  Basel, 1992.
\newblock [2012 reprint of the 1992 original] [MR1164866].

\bibitem{MR2662722}
K.~R. Parthasarathy.
\newblock What is a {G}aussian state?
\newblock {\em Commun. Stoch. Anal.}, 4(2):143--160, 2010.

\bibitem{MR710486}
A.~Pazy.
\newblock {\em Semigroups of linear operators and applications to partial
  differential equations}, volume~44 of {\em Applied Mathematical Sciences}.
\newblock Springer-Verlag, New York, 1983.

\bibitem{MR0385023}
W.~Rudin.
\newblock {\em Principles of mathematical analysis}.
\newblock McGraw-Hill Book Co., New York-Auckland-D\"usseldorf, third edition,
  1976.
\newblock International Series in Pure and Applied Mathematics.

\bibitem{MR2953553}
K.~Schm{\"u}dgen.
\newblock {\em Unbounded self-adjoint operators on {H}ilbert space}, volume 265
  of {\em Graduate Texts in Mathematics}.
\newblock Springer, Dordrecht, 2012.

\bibitem{MR4065489}
A.~E. Teretenkov.
\newblock Irreversible quantum evolution with quadratic generator: review.
\newblock {\em Infin. Dimens. Anal. Quantum Probab. Relat. Top.},
  22(4):1930001, 27, 2019.

\end{thebibliography}
\end{document}